\documentclass{article}

\usepackage[english]{babel}

\usepackage{arxiv}

\usepackage{amssymb}
\usepackage{times}
\usepackage{cite}

\usepackage{amsthm}
\usepackage{amsxtra}
\usepackage{amsmath}

\usepackage{verbatim}
\usepackage{epsfig}
\usepackage{setspace}
\usepackage{caption}
\usepackage{breqn}
\usepackage{graphicx}
\usepackage{tikz}
\usetikzlibrary{positioning}
\usepackage[colorlinks=true, allcolors=blue]{hyperref}

\newtheorem{theorem}{Theorem}
\newtheorem{definition}{Definition}
\newtheorem{lemma}[theorem]{Lemma}
\newtheorem{proposition}[theorem]{Proposition}
\newtheorem{corollary}[theorem]{Corollary}

\newtheorem{conjecture}{Conjecture}
\newtheorem{problem}{Problem}
\newtheorem{remark}{Remark}

\title{Equidistant Linear Codes in Projective Spaces\thanks{Content of this paper was partially covered in the doctoral thesis of the author at Indian Institute of Science, Bangalore.}}
\author{ \href{https://orcid.org/0000-0002-9110-7617}{\hspace{1mm}Pranab Basu}\\
	R. C. Bose Center for Cryptology and Security\\
	Indian Statistical Institute\\
	Kolkata 700108 \\
	\texttt{pranabbasu@alum.iisc.ac.in} \\}
\date{}

\renewcommand{\headeright}{}
\renewcommand{\undertitle}{}
\renewcommand{\shorttitle}{Equidistant Linear Codes in Projective Spaces}

\begin{document}
\maketitle

\begin{abstract}
Linear codes in the projective space $\mathbb{P}_q(n)$, the set of all subspaces of the vector space $\mathbb{F}_q^n$, were first considered by Braun, Etzion and Vardy. The Grassmannian $\mathbb{G}_q(n,k)$ is the collection of all subspaces of dimension $k$ in $\mathbb{P}_q(n)$. We study equidistant linear codes in $\mathbb{P}_q(n)$ in this paper and establish that the normalized minimum distance of a linear code is maximum if and only if it is equidistant. We prove that the upper bound on the size of such class of linear codes is $2^n$ when $q=2$ as conjectured by Braun et al. Moreover, the codes attaining this bound are shown to have structures akin to combinatorial objects, viz. \emph{Fano plane} and \emph{sunflower}. We also prove the existence of equidistant linear codes in $\mathbb{P}_q(n)$ for any prime power $q$ using \emph{Steiner triple system}. Thus we establish that the problem of finding equidistant linear codes of maximum size in $\mathbb{P}_q(n)$ with constant distance $2d$ is equivalent to the problem of finding the largest $d$-intersecting family of subspaces in $\mathbb{G}_q(n, 2d)$ for all $1 \le d \le \lfloor \frac{n}{2}\rfloor$. Our discovery proves that there exist equidistant linear codes of size more than $2^n$ for every prime power $q > 2$.
\end{abstract}

\keywords{Subspace codes \and Symmetric designs \and Intersecting families \and Linear codes \and Equidistant codes \and Steiner Triple System}
\setlength{\parindent}{10ex}
\section{Introduction}
\label{intro}

The unique finite field with $q$ elements is denoted by $\mathbb{F}_q$; $q$ is necessarily a prime power \cite{MS}. The \emph{projective space} $\mathbb{P}_q(n)$ is defined as the set of all subspaces of $\mathbb{F}_q^n$, the vector space of dimension $n$ over $\mathbb{F}_q$ \footnote{This terminology is not standard. In some branches of mathematics such as \emph{projective geometry}, the term projective space means the collection of all lines passing through the origin of a vector space.}. In terms of notation, $\mathbb{P}_q(n) := \{V: V \le \mathbb{F}_q^n\}$, where $\le$ signifies the usual vector space inclusion. The collection of all $k$-dimensional subspaces of $\mathbb{F}_q^n$ is known as the \emph{Grassmannian} of dimension $k$, denoted by $\mathbb{G}_q(n, k)$, for all nonnegative integers $k \le n$. In terms of notation we can write the following:
\begin{equation*}
    \mathbb{G}_q(n, k) := \{U \in \mathbb{P}_q(n): \dim U = k\}; \qquad \mathbb{P}_q(n) = \bigcup\limits_{i=0}^{n} \mathbb{G}_q(n, i).
\end{equation*}
The $q$-analog of the binomial coefficient is the $q$-ary \emph{Gaussian coefficient}, defined for all nonnegative integers $n$ and $k \le n$ as
\begin{equation*}
    {n \brack k}_q := |\mathbb{G}_q(n, k)| = \prod_{i=0}^{k-1} \frac{q^{n-i}-1}{q^{k-i}-1},
\end{equation*}
which expresses the total number of $k$-dimensional subspaces of $\mathbb{F}_q^n$. For all $X, Y \in \mathbb{P}_q(n)$, the \emph{subspace distance}, defined as
\begin{eqnarray*}
d_S (X, Y) &:=& \dim (X+Y) - \dim (X \cap Y) \\
&=& \dim X + \dim Y - 2\dim (X \cap Y),
\end{eqnarray*}
is a metric for $\mathbb{P}_q(n)$ \cite{AAK, KK}; where $X+Y$ is the smallest subspace that contains both $X$ and $Y$. Both $\mathbb{P}_q(n)$ and $\mathbb{G}_q(n, k)$ are turned into metric spaces under this metric.

A subset $\mathbb{C} \subseteq \mathbb{P}_q(n)$ of the projective space $\mathbb{P}_q(n)$ is called an $(n, M, d)$ \emph{code} in $\mathbb{P}_q(n)$ if the number of subspaces in $\mathbb{C}$ is $|\mathbb{C}| = M$ and $d_S (X, Y) \ge d$ for all distinct $X, Y \in \mathbb{C}$. The parameters $n, M$ and $d$ are the \emph{length}, \emph{size} and \emph{minimum distance} of $\mathbb{C}$, respectively. A code is \emph{equidistant} if any pair of codewords are at a fixed distance apart. The \emph{rate} $R$ and \emph{normalized minimum distance} $\delta$ of an $(n, M, d)$ code $\mathbb{C}$ are defined as follows.
\begin{equation*}
    l := \max\limits_{X \in \mathbb{C}} \dim X; \qquad R := \frac{\log_q {M}}{nl}; \qquad \delta := \frac{d}{2l}.
\end{equation*}
By definition both the rate and the normalized minimum distance of a code are limited to the range $[0, 1]$.

A code in a projective space is generally known as a \emph{subspace code}. A subspace code is called a \emph{constant dimension code} if all its codewords have a fixed nonzero dimension, i.e. $\mathbb{C} \subseteq \mathbb{P}_q(n)$ is a constant dimension code if $\mathbb{C} \subseteq \mathbb{G}_q(n, k)$ for some $k \le n$. For transmission of information in volatile networks, the technique of \emph{random network coding} is employed \cite{HKMKE, HMKKESL} where the intermediary nodes prefer coding of received inputs by taking random $q$-linear combinations instead of simple routing. In that context, the authors of \cite{KK} showed that subspace codes can tackle the problem of errors and erasures introduced anywhere in a network with varying topology. This motivated a vast amount of research in codes in projective spaces \cite{EV, SE, SE2, HKK, GY, XF, TMBR, KoK}.

A partially ordered set is called a \emph{lattice} if a least upper bound and a greatest lower bound of any pair of elements exist within the set. The lattice associated to $\mathbb{P}_q(n)$ is the \emph{projective lattice}, denoted by $(\mathbb{P}_q(n), +, \cap, \le)$. Within a framework of lattices, codes in projective spaces are the $q$-analog of binary block codes in Hamming spaces \cite{BEV}. While linear codes in Hamming spaces have been extensively studied due to their application in error correction \cite{MS, Ru}, not much attention was given to their counterparts in projective spaces until Braun et al. introduced the notion of linearity in $\mathbb{P}_q(n)$ \cite{BEV}. An $[n, k]$ binary linear block code in $\mathbb{F}_2^n$ is a subspace of dimension $k$. However, the extension of this definition to projective spaces is not straightforward. The subspace distance is not \emph{translation invariant} over entire $\mathbb{P}_q(n)$ as $\mathbb{P}_q(n)$ is not a vector space with respect to the usual vector space addition. This problem was tackled in \cite{BEV} by imposing a structure akin to a vector space to a subset of $\mathbb{P}_q(n)$. They conjectured the following about the size of such a subset wherein linearity can be defined.
\begin{conjecture}
\label{C}
The maximum size of a linear code in $\mathbb{P}_2(n)$ is $2^n$.
\end{conjecture}

The extension of Conjecture~\ref{C} to linear codes in $\mathbb{P}_q(n)$ for any prime power $q$ has been proved for special cases in \cite{PS, BK2}. The authors of \cite{BK} showed that a linear code in $\mathbb{P}_q(n)$ can be as large as $2^n$ if it is a sublattice of the corresponding projective lattice; furthermore, the maximal code has a normalized minimum distance of $\frac{1}{2n}$, the lowest among all $(n, M, d)$ codes. Thus, it is important to look beyond lattice scheme for linear subspace codes that can achieve higher rate with higher normalized minimum distance.

Coding in projective spaces is tricky as $\mathbb{P}_q(n)$ is not \emph{distance-regular}; therefore, standard geometric intuitions often do not hold. This problem is overcome for constant dimension codes as the underlying coding space $\mathbb{G}_q(n, k)$ is distance-regular. In this context, study of equidistant codes in Grassmannians have been motivated as they exhibit very unique distance-distribution \cite{ER, BP, GR, BSSV}. An equidistant code in $\mathbb{G}_q(n, k)$ is necessarily a $\lambda$-intersecting code, i.e. any two codewords always intersect in a subspace of some fixed dimension $\lambda < k$. The special case when all pairwise intersections coincide is known as a \emph{sunflower}.

Equidistant constant dimension codes are the $q$-analog of equidistant constant-weight codes in Hamming spaces. In classical error correction, the maximum size of an equidistant code in $\mathbb{F}_q^n$ is closely related to that of an equidistant constant-weight code in $\mathbb{F}_q^n$ with same constant distance and constant weight \cite{FKLW}. We use the notation $\mathbb{B}_q(n, d)$ and $\mathbb{B}_q(n, d, w)$ to denote the maximum size of an equidistant code with constant distance $d$ and that of an equidistant constant-weight code with constant distance $d$ and constant weight of $w$, respectively.
\begin{theorem}[\cite{FKLW}, Theorem~1]
    \label{A}
    \begin{equation*}
    \mathbb{B}_q(n, d) = \mathbb{B}_q(n, d, d) + 1.
    \end{equation*}
\end{theorem}
An equivalent expression for equidistant constant dimension codes may exist. In the sequel, we will use the term `linear code' to imply linear subspace code unless mentioned otherwise.

In this paper, we investigate equidistant linear codes in projective spaces. In particular, we show that the nontrivial part of an equidistant linear code in $\mathbb{P}_q(n)$ is basically an equidistant code in $\mathbb{G}_q(n, 2d)$ with constant distance of $2d$ for some $d \le{\lfloor \frac{n}{2}\rfloor}$. This helps us to bring out the $\lambda$-intersecting structure in such class of linear codes. As a result of this characterization, we prove the conjectured bound (Conjecture~\ref{C}) for equidistant linear codes in $\mathbb{P}_2(n)$ and construct maximal codes that resemble well known combinatorial objects like Fano plane and sunflowers. We also show that constructing equidistant linear codes in a projective space is equivalent to finding certain intersecting families in the Grassmannians contained therein.

The rest of the paper is organized as follows. In Section~\ref{sec:1} we give the formal definition of a linear code in a projective space and some pertinent definitions of certain combinatorial objects. A few useful results are also listed. Equidistant linear codes are formally defined and discussed in Section~\ref{sec:4}, where we prove that the nontrivial part of an equidistant linear code in $\mathbb{P}_q(n)$ with constant distance $2d$ is a $d$-intersecting code in $\mathbb{G}_q(n, 2d)$. A linear code is shown to attain the maximum normalized minimum distance if and only if it is equidistant. In Section~\ref{sec:5} we prove the upper bound conjectured by Braun et al. on the size of a linear code in $\mathbb{P}_2(n)$ (Conjecture~\ref{C}) when the code is equidistant. The bound is shown to be achievable only for the case $n=3$ using results from design theory. We present constructions for the maximal codes for all $n \ge 3$; the maximal codes turn out to be either a Fano plane or sunflowers. Equidistant linear codes in $\mathbb{P}_q(n)$ for any prime power $q$ are considered in Section~\ref{sec:6}, where using Steiner triple systems we establish that existence of an equidistant linear code in $\mathbb{P}_q(n)$ with constant distance $2d$ is guaranteed if and only if there exists a $d$-intersecting family in $\mathbb{G}_q(n, 2d)$ of similar size. Thereby we prove that translation invariance on the subspace distance metric is decoupled by the property of constant distance in linear codes in projective spaces. Finally, equidistant linear codes of size greater than $2^n$ in $\mathbb{P}_q(n)$ are shown to exist for all prime powers $q >2$. We discuss a few interesting open problems for future research in Section~\ref{sec:7}.
\paragraph{Notation.} The finite vector space of dimension $n$ over $\mathbb{F}_q$ is denoted by $\mathbb{F}_q^n$. $\mathbb{P}_q(n)$ denotes the set of all subspaces of $\mathbb{F}_q^n$. We denote the canonical $m$-set as $[m] := \{1, \ldots, m\}$. $\mathcal{P}(\mathcal{X})$ and $\binom{\mathcal{X}}{k}$ represent the collection of all subsets and that of all $k$-subsets of a finite set $\mathcal{X}$, respectively. A Steiner triple system defined on a set of $v$ points will be written as $STS(v)$. The maximum size of an equidistant constant dimension code in $\mathbb{G}_q(n, k)$ with constant distance $d$ will be expressed as $\mathcal{A}_q(n, d, k)$. $\mathcal{U}^{*}$ will stand for the nontrivial part of a linear code $\mathcal{U}$, i.e. $\mathcal{U}^{*} := \mathcal{U} \backslash \{\{0\}\}$. The subspace spanned by a set of vectors $\{v_1, \ldots, v_i\}$ will be denoted by $\langle v_1, \ldots, v_i\rangle$. The all-zero vector will be denoted as $\mathbf{0} := (0, \ldots, 0)$. For any real number $x$, the greatest integer less than $x$ is represented by $\lfloor x\rfloor$.

\section{Definitions and Relevant Background}
\label{sec:1}
In this section we will go through some basic definitions as well as a few well known results that can be found in the literature. The notation and terminology used here are standard.
\subsection{Linear Codes in Projective Spaces}
\label{sec:2}
A linear code in the projective space $\mathbb{P}_q(n)$ is a subset of $\mathbb{P}_q(n)$ containing $\{0\}$ that satisfies certain criteria as stated below.
\begin{definition}
	\label{L}
	A subset $\mathcal{U}  \subseteq \mathbb{P}_q(n)$, with $\left\{0\right\} \in \mathcal{U}$, is a \emph{linear code} if there exists an addition operation $\boxplus : \mathcal{U} \times \mathcal{U} \rightarrow \mathcal{U}$ such that: 
	\begin{itemize}
		\item [(i)] $(\mathcal{U}, \boxplus)$ is an abelian group; 
		\item [(ii)] the identity element of the group $(\mathcal{U}, \boxplus)$ is $\left\{0\right\}$; 
		\item [(iii)] $X \boxplus X = \left\{0\right\}$ for every group element $X \in \mathcal{U}$; 
		\item [(iv)] the function $\boxplus$ is isometric, i.e., $d_S(Y_1 \boxplus X, Y_2 \boxplus X) = d_S(Y_1, Y_2)$ for all $Y_1, Y_2, X \in \mathcal{U}.$
	\end{itemize}
\end{definition}

Braun et. al. introduced the above definition to define linear codes in projective spaces when the underlying field is binary, i.e. when $q = 2$ \cite{BEV}. Observe that $0.X = \{0\}$ and $1.X = X$ for all $X \in \mathcal{U}$, i.e., the scalar multiplication is distributive over $\boxplus$ in $\mathbb{F}_2$. Definition \ref{L} was later generalized to any prime power $q$ while keeping the self-inversion condition (Definition \ref{L}(iii)) intact \cite{PS, BK2, BK}, thus guaranteeing distributivity for the scalar multiplication. The consequence is that any linear code $\mathcal{U}$ in $\mathbb{P}_q(n)$ is a vector space with respect to the corresponding addition operation $\boxplus$ over $\mathbb{F}_2$ \cite{BEV, BTh}.

Linearity imposes certain geometric restrictions on any linear code in $\mathbb{P}_q(n)$, some of which are encapsulated in the following three lemmas \cite{BEV}. We record them without the proofs.
\begin{lemma}(\cite{BEV}, Lemma~6)
	\label{L1}
	Let $\mathcal{U}$ be a linear code in $\mathbb{P}_q(n)$ and let $\boxplus$ be the isometric linear addition on $\mathcal{U}$. Then for all $X, Y \in \mathcal{U}$, we have:
	\begin{equation*}
	\dim(X \boxplus Y) = d_S(X, Y) = \dim X + \dim Y - 2\dim(X \cap Y)
	\end{equation*}
	In particular, if $X \subseteq Y$, then $\dim(X \boxplus Y) = \dim Y - \dim X$.
\end{lemma}
\begin{lemma}(\cite{BEV}, Lemma~7)
	\label{L2}
	For any three subspaces $X, Y$ and $Z$ of a linear code $\mathcal{U}$ in $\mathbb{P}_q(n)$ with isometric linear addition $\boxplus$, the condition $Z = X \boxplus Y$ implies $Y = X \boxplus Z$.
\end{lemma}
\begin{lemma}(\cite{BEV}, Lemma~8)
	\label{L3}
	Let $\mathcal{U}$ be a linear code in $\mathbb{P}_q(n)$ and let $\boxplus$ be the isometric linear addition on $\mathcal{U}$. If $X$ and $Y$ are any two codewords of $\mathcal{U}$ such that $X \cap Y = \left\{0\right\}$, then $X \boxplus Y = X + Y$. Also $\dim(X \boxplus Y) = \dim X + \dim Y$.
\end{lemma}
The minimum distance of a linear code in $\mathbb{P}_q(n)$ is the minimum dimension of a nontrivial codeword as illustrated next.
\begin{lemma}
	\label{L4}
	If $\mathcal{U} \subseteq \mathbb{P}_q(n)$ is a linear code then the minimum distance $d$ of $\mathcal{U}$ is $d = \min\limits_{X \in \mathcal{U}^{*}} \dim X$.
\end{lemma}
\begin{proof}
	Consider any $Z \in \mathcal{U}^{*}$ such that $\dim Z \le \dim Y$ for all $Y \in \mathcal{U}^{*}$. By Lemma~\ref{L1} and minimality of $d$, $d \le d_S (Z, \{0\}) = \dim Z$. Suppose $d < \dim Z$. There must exist $Z_1, Z_2 \in \mathcal{U}$ such that $d = d_S(Z_1, Z_2)$. Our claim is proved since $\dim (Z_1 \boxplus Z_2) = d_S(Z_1, Z_2) < \dim Z$ is a contradiction. 
\end{proof}
To understand how codewords of linear codes in projective spaces intersect with each other, the next lemma from \cite{BK} is crucial.
\begin{lemma}(\cite{BK}, Lemma~5)
	\label{L5}
		If $\mathcal{U} \subseteq \mathbb{P}_q(n)$ is a linear code and $X, Y \in \mathcal{U}$, then
	\begin{itemize}
		\item[(i)] $\dim X = \dim (X \cap Y) + \dim(X \cap (X \boxplus Y))$;
		\item[(ii)] $\dim(X \boxplus Y) = \dim(X \cap (X \boxplus Y)) + \dim(Y \cap (X \boxplus Y))$.
	\end{itemize}
\end{lemma}
We record two more results about linearity in projective spaces \cite{BK}.
\begin{lemma}(\cite{BK}, Lemma~6)
	\label{L6}
	Let $\mathcal{U} \subseteq \mathbb{P}_q(n)$ be a linear code. For all $X, Y \in \mathcal{U}$ the following is true:
	\begin{equation*}
	\dim (X \boxplus Y) \geq \dim X - \dim Y,
	\end{equation*}
	with equality if and only if $Y \subseteq X$.
\end{lemma}
\begin{lemma}(\cite{BK}, Lemma~7)
	\label{L7}
	Let $\mathcal{U}$ be a linear code in $\mathbb{P}_q(n)$ and $X, Y$ be two distinct nontrivial codewords of $\mathcal{U}$. Then,
	\begin{itemize}
		\item[(a)] $Y \subset X$ if and only if $(X \boxplus Y) \subset X$;
		\item[(b)] $Y \subset X$ if and only if $Y \cap (X \boxplus Y) = \{0\}$.
	\end{itemize}
\end{lemma}

In the subsequent sections we will make use of the above lemmas to bring out the intersecting properties of equidistant linear codes in projective spaces.
\subsection{Combinatorial Structures}
\label{sec:3}

We will briefly go through some basic definitions regarding a few combinatorial objects and some well known results.
\begin{definition}
	\label{D1}
	A $t$-$(v, k, \lambda)$ \emph{design} over a finite $v$-set $\mathcal{X}$ is a family $\mathcal{D}$ of distinct $k$-subsets of $\mathcal{X}$, called \emph{blocks}, such that each $t$-subset in $\binom{\mathcal{X}}{t}$ is contained in exactly $\lambda$ blocks of $\mathcal{D}$.
\end{definition}
\begin{definition}
	\label{D2}
	For any positive integer $\lambda \in \mathbb{Z}^{+}$, a $\lambda$-intersecting family $\mathcal{F} \subset \mathcal{P}(\mathcal{X})$ defined on a finite set $\mathcal{X}$ is a \emph{symmetric design} if $\mathcal{F} \subseteq \binom{\mathcal{X}}{k}$ for some $k \in \mathbb{Z}^{+}$ and $|\mathcal{F}| = |\mathcal{X}|$.
\end{definition}
The \emph{degree} of an element $x \in \mathcal{X}$ in a family $\mathcal{F} \subset \mathcal{P}(\mathcal{X})$, denoted by $\deg(x)$, is the number of members in $\mathcal{F}$ containing $x$: $\deg(x) := |\{A \in \mathcal{F}: x \in A\}|$. $\mathcal{F} \subset \mathcal{P}(\mathcal{X})$ is an \emph{$r$-regular} family if $\deg(x) = r$ for all $x \in \mathcal{X}$, $r$ being the \emph{replication number} of $\mathcal{F}$.

We will henceforth use the notation $(v, k, \lambda)$ design to imply a $2$-$(v, k, \lambda)$ design. Ryser proved that symmetric designs are particular cases of designs when the number of blocks and points are the same \cite{R}.
\begin{theorem}(\cite{ANC}, Theorem~3.1.4)
	\label{T1}
	For some positive integer $\lambda \in \mathbb{Z}^{+}$, if a $\lambda$-intersecting family $\mathcal{F} \subset \binom{\mathcal{X}}{k}$ is a symmetric design, then $\mathcal{F}$ is $k$-regular and any pair of distinct points $\{x, y\} \in \binom{\mathcal{X}}{2}$ in $\mathcal{X}$ is contained in exactly $\lambda$ blocks of $\mathcal{F}$.
\end{theorem}

The parameters of a design are known to be related to each other in the following way.
\begin{theorem}(\cite{J}, Theorem~12.2)
	\label{T2}
	If $\mathcal{D}$ is an $r$-regular $(v, k, \lambda)$ design containing $b$ blocks then
	\begin{eqnarray}
	r(k-1) &=& \lambda(v-1); \nonumber \\
	bk &=& vr. \nonumber
	\end{eqnarray}
\end{theorem}

The particular case of a $(v, k, \lambda)$ design with parametric values $k = 3$ and $\lambda = 1$ is known as a \emph{Steiner triple system}, denoted as $S(2, 3, v)$ or $STS(v)$ \cite{S}. A formal definition follows.
\begin{definition}
	A Steiner triple system $S(2, 3, v)$ defined on a finite $v$-set $\mathcal{X}$ is a collection of distinct $3$-subsets of $\mathcal{X}$, known as \emph{blocks}, such that every pair of distinct points in $\mathcal{X}$ is contained in exactly one block.
\end{definition}

Kirkman showed that the existence of an $STS(v)$ depends solely on the positive integer $v \in \mathbb{Z}^{+}$ \cite{K}.
\begin{theorem}(\cite{K})
	\label{T3}
	An $STS(v)$ exists if and only if $v \in \mathbb{Z}^{+}$ satisfies the congruence relation $v \equiv 1, 3 \mod 6$.
\end{theorem}

Another particular example of $(v, k, \lambda)$ designs frequently used as combinatorial objects is \emph{projective planes}, which can be defined in the following way.
\begin{definition}
	For a positive integer $p \in \mathbb{Z}^{+}$, a projective plane of order $p$ consists of a finite $(p^2 + p + 1)$-set $\mathcal{X}$ and a family $\mathcal{F} \subset \binom{\mathcal{X}}{p+1}$ of $(p+1)$-subsets of $\mathcal{X}$, called \emph{lines}, such that any pair of distinct points of $\mathcal{X}$ lie on a unique line in $\mathcal{F}$.
\end{definition}
A projective plane of order $p$ also contains $p^2+p+1$ lines \cite[Proposition~12.9]{J}. Thus, in other words, a projective plane of order $p$ is a $(p^2+p+1, p+1, 1)$ symmetric design. The only example of a projective plane of order $p = 1$ is the \emph{triangle}, while the unique projective plane of order $p=2$ is the \emph{Fano plane} (Fig.~\ref{F1}).

\begin{figure}[h!]
	\centering
	\includegraphics[width=0.3\linewidth]{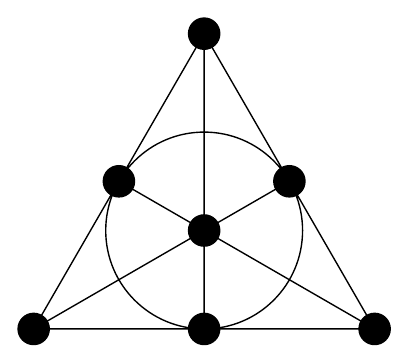}
	\caption{$(7, 3, 1)$ symmetric design: Fano plane}
	\label{F1}
\end{figure}



We now give a formal definition to the $q$-analog of $\lambda$-intersecting families of sets.
\begin{definition}
	A family of subspaces $\mathcal{F} \in \mathbb{P}_q(n)$ is $\lambda$-intersecting for some positive integer $\lambda \in \mathbb{Z}^{+}$ if $\dim (X \cap Y) = \lambda$ for all distinct members $X, Y$ of $\mathcal{F}$.
\end{definition}
Now onwards we will say that a linear code $\mathcal{U} \subseteq \mathbb{P}_q(n)$ is $\lambda$-intersecting if its nontrivial part $\mathcal{U}^{*}$ is a $\lambda$-intersecting family of subspaces in $\mathbb{P}_q(n)$. At this point we record the well known Fisher's inequality \cite{F, M} that puts a tight upper bound on the size of any $\lambda$-intersecting family defined on a finite set.
\begin{theorem}(\cite{ANC}, Theorem~2.3.9)
	\label{T4}
	If $\mathcal{F} \subset \mathcal{P}(\mathcal{X})$ is a $\lambda$-intersecting family on the finite set $\mathcal{X}$ for some positive integer $\lambda \in \mathbb{Z}^{+}$, then $|\mathcal{F}| \le |\mathcal{X}|$. 
\end{theorem}
However, it must be noted that the $q$-analog of Theorem~\ref{T4} is yet unknown.

It is not clear whether any $\lambda$-intersecting family defined on a finite set attaining the upper bound stated in Theorem~\ref{T4} must have a specific structure for all permissible values of $\lambda \in \mathbb{Z}^{+}$. For the particular case of $\lambda =1$, however, all types of extremal families are investigated by De Bruijn and Erd\"{o}s \cite{DBE}. We present only the case under the additional constraint of the family being uniform.
\begin{theorem}(\cite{DBE}, Theorem~1)
	\label{T5}
	Let $\mathcal{F} \subset \mathcal{P}(\mathcal{X})$ be a $1$-intersecting family defined on the finite set $\mathcal{X}$. If $|\mathcal{F}| = |\mathcal{X}|$ and $\mathcal{F} \subset \binom{\mathcal{X}}{k}$ for some positive integer $k \in \mathbb{Z}^{+}$, then $\mathcal{F}$ is a projective plane.
\end{theorem}

Members of a $\lambda$-intersecting family of subspaces always intersect pairwise in a subspace of dimension $\lambda$. When all members of a $\lambda$-intersecting uniform family also intersect mutually in a subspace of dimension $\lambda$, the family is known as a \emph{sunflower}.
\begin{definition}
	A $t$-intersecting family of $k$-subspaces $\mathbb{S} \subset \mathbb{G}_q(n, k)$ is a sunflower if $X \cap Y = C$ for all distinct $X, Y \in \mathbb{S}$ and some fixed $t$-subspace $C \in \mathbb{G}_q(n, t)$. $C$ is called the \emph{center} of the sunflower $\mathbb{S}$.
\end{definition}

The maximum size of a sunflower formed by subsets of a finite set is given by the widely known Erd\"{o}s-Ko-Rado theorem \cite{EKS}. The $q$-analog of the same also exists as a consequence of the work done in \cite{H, FW, GN, N}.
\begin{theorem}(\cite{H}, \cite{FW}, \cite{GN}, \cite{N})
	\label{T6}
	Let $\mathcal{F} \subset \mathbb{G}_q(n, k)$ be an intersecting family of $k$-subspaces, where $2k \le n$. Then $|\mathcal{F}| \le {{n-1} \brack {k-1}}_{q}$, and equality holds if and only if
	\begin{itemize}
		\item [(a)] $\mathcal{F} = \{V \in \mathbb{G}_q(n, k): W \subset V\}$ for some $W \in \mathbb{G}_q(n, 1)$, or
		\item [(b)] $n = 2k$ and $\mathcal{F} = \{Y \in \mathbb{G}_q(2k, k): Y \subset H)\}$ for some $H \in \mathbb{G}_q(2k, 2k-1)$.
	\end{itemize}
\end{theorem}
\section{Equidistant Linear Codes in $\mathbb{P}_q(n)$}
\label{sec:4}

An \emph{equidistant linear code} in the projective space $\mathbb{P}_q(n)$ is simply a linear code in $\mathbb{P}_q(n)$ that is equidistant.
\begin{definition}
	A linear code $\mathcal{U} \subseteq \mathbb{P}_q(n)$ is equidistant with constant distance $r$ if $d_S (X, Y) = r$ for all distinct codewords $X, Y \in \mathcal{U}$.
\end{definition}

Our goal is to show that equidistant linear codes in projective spaces are analogous to $\lambda$-intersecting linear codes for suitable values of $\lambda \in \mathbb{Z}^{+}$. In the process we will establish that the nontrivial part of any such linear code is uniform, i.e. it contains only subspaces of a fixed dimension. Moreover, the constant dimension is same as the constant distance. We begin with the proof that the distance of an equidistant linear code is permitted to be only an even integer.
\begin{lemma}
	\label{L8}
	The constant distance of an equidistant linear code $\mathcal{U}$ is always even.
\end{lemma}
\begin{proof}
	Suppose $r \in \mathbb{Z}^{+}$ is the constant distance of $\mathcal{U}$. Then $r = d_S (X, Y) = \dim X + \dim Y - 2\dim (X \cap Y)$ for some distinct $X, Y \in \mathcal{U}^{*}$ and the statement is proved since $\dim X = d_S (X, \{0\}) = d_S (Y, \{0\}) = \dim Y$. 
\end{proof}
\begin{theorem}
	\label{T7}
	If $\mathcal{U} \subseteq \mathbb{P}_q(n)$ is a linear code, then the following statements are equivalent:
	\begin{itemize}
		\item [(a)] $\mathcal{U}$ is equidistant with constant distance $2d$.
		\item [(b)] $\mathcal{U}^{*} \subseteq \mathbb{G}_q(n, 2d)$.
		\item [(c)] $\mathcal{U}^{*}$ is a $d$-intersecting family.
	\end{itemize}
\end{theorem}
\begin{proof}
	((a) $\Leftrightarrow$ (b)) Suppose $\mathcal{U}$ is equidistant with constant distance $2d \in \mathbb{Z}^{+}$. For any nontrivial codeword $X \in \mathcal{U}^{*}$,  $d_S(X, \{0\}) = 2d$, which implies that $\dim X = 2d$, i.e. $\mathcal{U}^{*} \subseteq \mathbb{G}_q(n, 2d)$.
	Conversely, $d_S (X, Y) = \dim(X \boxplus Y) = 2d$ for all distinct $X, Y \in \mathcal{U}$ if $\mathcal{U}^{*} \subseteq \mathbb{G}_q(n, 2d)$.
	
	((b) $\Leftrightarrow$ (c)) If $X, Y \in \mathcal{U}^{*} \subseteq \mathbb{G}_q(n, 2d)$ are distinct, we have $2d = \dim (X \boxplus Y) = 4d - 2\dim(X \cap Y)$ by Lemma~\ref{L1}, i.e. $\dim (X \cap Y) = d$.
	
	Conversely, say $\mathcal{U}^{*}$ is $d$-intersecting. $X \boxplus Y \ne \{0\}$ for any distinct $X, Y \in \mathcal{U}^{*}$; thus according to Lemma~\ref{L5} we can write: $\dim X = \dim (X \cap Y) + \dim (X \cap (X \boxplus Y)) = 2d$. Hence, proved. 
\end{proof}
\begin{remark}
	\label{R1}
It must be noted that no linear code $\mathcal{U} \subseteq \mathbb{P}_q(n)$ can have a nontrivial part $\mathcal{U}^{*} \subseteq \mathbb{G}_q(n, r)$ for some odd integer $r \in \mathbb{Z}^{+}$. Otherwise, $d_S (X, Y) = \dim (X \boxplus Y) = r$ for any pair of codewords $X, Y \in \mathcal{U}$, which indicates that $\mathcal{U}$ is equidistant with an odd constant distance $r$, a contradiction to Lemma~\ref{L8}.
\end{remark}
\begin{corollary}
	\label{C1}
	The largest normalized minimum distance of a linear code $\mathcal{U} \subseteq \mathbb{P}_q(n)$ is $\frac{1}{2}$ and it is achieved if and only if $\mathcal{U}$ is equidistant.
\end{corollary}
\begin{proof}
	If $\delta$ is the normalized minimum distance of $\mathcal{U}$ then by Lemma~\ref{L4} and definition of $\delta$,
	\begin{equation*}
	\delta = \frac{\min \limits_{X, Z \in \mathcal{U}^{*}, X \ne Z} d_S (X, Z)}{2 \max \limits_{Y \in \mathcal{U}^{*}} \dim Y} = \frac{\min \limits_{X \in \mathcal{U}^{*}} \dim X}{2 \max \limits_{Y \in \mathcal{U}^{*}} \dim Y} \le \frac{1}{2}.
	\end{equation*}
	Clearly, $\delta = \frac{1}{2}$ if and only if $\mathcal{U}^{*} \subseteq \mathbb{G}_q(n, r)$ for some positive integer $r \in \mathbb{Z}^{+}$. The result then follows from Theorem~\ref{T7}. 
\end{proof}

Thus, the property of constant distance helps codes in projective spaces acquire a uniquely beneficial distance-distribution within the constraints of linearity. Theorem~\ref{T7} helps to bring out the $\lambda$-intersecting structure of equidistant linear codes within a Grassmannian, and we will henceforth follow that structure. We end this section with a final observation.
\begin{remark}
	\label{R2}
	Theorem~\ref{T7} implicits that the size of an equidistant linear code in $\mathbb{P}_q(n)$ is upper bounded by $$\max\limits_{1 \le d \le \lfloor \frac{n}{2} \rfloor} \mathcal{A}_q(n, 2d, 2d) +1,$$ which is analogous to Theorem~\ref{A}.
\end{remark}
\section{Constructions and Bounds for Equidistant Linear Codes in $\mathbb{P}_2(n)$}
\label{sec:5}
Linear codes in $\mathbb{P}_2(n)$ are conjectured to have a maximum size of $2^n$ in \cite{BEV}. We will prove that the conjectured bound holds when a code has an additional constraint of constant distance. First, let us go through the only example of an equidistant linear code in the literature which attains this bound. Braun et al. constructed the following code without bringing in the notion of equidistant code \cite{BEV}.
\paragraph{Construction: Equidistant linear code for $n=3, q =2$}
\begin{enumerate}
	\item Choose $\alpha$, a primitive element of $\mathbb{F}_2^3$ such that $\alpha^3 + \alpha + 1 = 0$.
	\item Define a code $\mathcal{U} \subseteq \mathbb{P}_2(n)$ as $\mathcal{U} := \{\{0\}\} \cup \{\langle \alpha^i, \alpha^{i+1}\rangle: 0\le i \le 6\}$.
	\item Define a commutative function $\boxplus : \mathcal{U} \times \mathcal{U} \longrightarrow \mathcal{U}$ as
	\begin{align*}
	\{0\} \boxplus \{0\} &:= \{0\}; \\
	\langle \alpha^i, \alpha^{i+1}\rangle \boxplus \{0\} &:= \langle \alpha^i, \alpha^{i+1}\rangle &\quad \text{for} \ 0 \le i \le 6; \\
	\langle \alpha^i, \alpha^{i+1}\rangle \boxplus \langle \alpha^j, \alpha^{j+1}\rangle &:= \langle \alpha^i + \alpha^j, \alpha^{i+1}+ \alpha^{j+1}\rangle &\quad \text{for} \ 0 \le i, j \le 6. 
	\end{align*}
	\item Note that $\mathcal{U}$ is a linear code in $\mathbb{P}_2(3)$ with size $|\mathcal{U}| = 2^3$; $\boxplus$ acts as a linear addition on $\mathcal{U}$.
	\item $\mathcal{U}$ is equidistant with a constant distance of $2$.
\end{enumerate}
It is easy to observe that the collection of blocks $\mathcal{S} := \{X \backslash \{\mathbf{0}\}: X \in \mathcal{U}^{*}\}$ defined on the finite set $\mathbb{F}_2^3 \backslash \{\mathbf{0}\}$ is the Fano plane. As a matter of fact, the above constructed code is the only example of an equidistant linear code in $\mathbb{P}_2(3)$ with size $2^3$.
\begin{proposition}
	\label{P1}
	There is only one equidistant linear code in $\mathbb{P}_2(3)$ with size of $2^3$.
\end{proposition}
\begin{proof}
	For any equidistant linear code $\mathcal{U} \subseteq \mathbb{P}_2(3)$, the collection of subsets $\mathcal{S} := \{X \backslash \{\mathbf{0}\}: X \in \mathcal{U}^{*}\}$ forms a uniform $1$-intersecting family on $\mathbb{F}_2^3 \backslash \{\mathbf{0}\}$. If $|\mathcal{S}| = 7$, $\mathcal{S}$ must be a projective plane of order $2$ by Theorem~\ref{T5}. The statement holds since there exists only one such projective plane. 
\end{proof}
\begin{remark}
	\label{R3}
	The linear addition $\boxplus$ in the construction presented above can also be defined in other ways. We will illustrate the same in Section~\ref{sec:6}.
\end{remark}

We will show that nontrivial equidistant linear codes in $\mathbb{P}_2(n)$ exist for all $n \ge 3$. However, the upper bound conjectured by Braun et al. is never reached for $n > 3$ as proved next.
\begin{theorem}
	\label{T8}
	If $\mathcal{U} \subseteq \mathbb{P}_2(n)$ is an equidistant linear code, then $|\mathcal{U}| \le 2^n$. The bound is reached only if $n=3$.
\end{theorem}
\begin{proof}
	According to Theorem~\ref{T7} the family of subsets $\mathcal{S} := \{X \backslash \{\mathbf{0}\}: X \in \mathcal{U}^{*}\}$ defined on $\mathbb{F}_2^n \backslash \{\mathbf{0}\}$ is $(2^d -1)$-intersecting  for some positive integer $d \in \mathbb{Z}^{+}$. Therefore, by Theorem~\ref{T4}, $|\mathcal{U}| = |\mathcal{S}|+1 \le 2^n$.
	
	Suppose $\mathcal{U}$ meets this bound. Since $2d$ is the constant distance of $\mathcal{U}$, the family $\mathcal{S}$ defined on the underlying set $\mathbb{F}_2^n \backslash \{\mathbf{0}\}$ is uniform and $(2^d-1)$-intersecting such that $|\mathcal{S}| = 2^n -1 = |\mathbb{F}_2^n \backslash \{\mathbf{0}\}|$; thus $\mathcal{S}$ is a $(2^n-1, 2^{2d}-1, 2^d-1)$ symmetric design. By Theorem~\ref{T1}, $\mathcal{S}$ is $(2^{2d}-1)$-regular.Therefore, Theorem~\ref{T2} yields: $(2^{2d}-1)(2^{2d}-2) = (2^d-1)(2^n-2)$. Upon simplifying we obtain
	\begin{equation}
	\label{E1}
	2^{n-d-1} = 2^{2d-1} + 2^{d-1} -1.
	\end{equation}
	Both sides of \eqref{E1} should be even since $n \ge 2d+1 > d+1$. Thus we must have $2^{d-1}-1 =0$, i.e. $d=1$, implying that the only solution to \eqref{E1} is $n=3$. 
\end{proof}

Linear codes in $\mathbb{P}_q(n)$ defined in Definition~\ref{D1} form vector spaces over $\mathbb{F}_2$ irrespective of the value of $q$. Consequently, size of equidistant linear codes in $\mathbb{P}_2(n)$ can be as large as $2^{n-1}$. A construction of such class of codes is demonstrated next for all $n \ge 3$.

\paragraph{Construction: Equidistant linear code for $n \ge 3, q=2$}
\begin{enumerate}
	\item Fix a basis $\mathcal{B}=\{u_0, u_1, \ldots, u_{n-1}\}$ of $\mathbb{F}_2^n$ over $\mathbb{F}_2$.
	\item Define a code $\mathcal{U} \subseteq \mathbb{P}_2(n)$ as $\mathcal{U} := \{\{0\}\} \cup \{\langle u_0, v\rangle: v \in \langle u_1, \ldots, u_{n-1}\rangle\}$.
	\item Define a linear addition $\boxplus: \mathcal{U} \times \mathcal{U} \longrightarrow \mathcal{U}$ as follows:
	\begin{align*}
	X \boxplus X &:= \{0\} &\quad \forall X \in \mathcal{U}; \\
	X \boxplus \{0\} &:= X &\quad \forall X \in \mathcal{U}; \\
	\langle u_0, v\rangle \boxplus \langle u_0, w\rangle &:= \langle u_0, v+w\rangle &\quad \forall v, w \in \langle u_1, \ldots, u_{n-1}\rangle \ni v \ne w.
	\end{align*}
	\item $\mathcal{U}^{*}$ is the collection of all $2$-dimensional subspaces of $\mathbb{F}_2^n$ containing $\langle u_0\rangle$, i.e. $\mathcal{U^{*}} = \{Z: Z \in \mathbb{G}_2(n, 2), u_0 \in Z\}$ is the largest sunflower in $\mathbb{G}_2(n, 2)$ with $\langle u_0\rangle$ as its center. Evidently, $|\mathcal{U}^{*}| = {{n-1} \brack {1}}_2 = 2^{n-1}-1$.
	\item $\mathcal{U}$ is an equidistant linear code in $\mathbb{P}_2(n)$ of size $2^{n-1}$ with constant distance of $2$.
\end{enumerate}
As is the case with construction of equidistant linear code in $\mathbb{P}_2(3)$ presented in \cite{BEV}, the linear addition $\boxplus$ in the above construction can be defined in different ways.

The above approach of constructing an equidistant linear code adopts a method of constructing a sunflower and appending $\{0\}$ as a codeword to it. We argue that this is the only way of constructing equidistant linear codes in $\mathbb{P}_2(n)$ with constant distance of $2$ and size of $2^{n-1}$ for all $n > 4$. However, there is another way to construct such a code for $n=4$: Fix a subspace $T$ of $\mathbb{F}_2^4$ such that $\dim T = 3$. The code consists of all $2$-dimensional subspaces of $\mathbb{F}_2^4$ contained in $T$ as well as $\{0\}$. Clearly the size of this code is ${{3} \brack {2}}_2 + 1 = 2^3$. The linear addition can be defined in a similar way used by Braun et al.
\begin{proposition}
	\label{P2}
	Let $\mathcal{U} \subseteq \mathbb{P}_2(n)$ be an equidistant linear code of constant distance $2$. For $n \ge 4, |\mathcal{U}| = 2^{n-1}$ only if $\mathcal{U}^{*}$ is a sunflower or $\mathcal{U}^{*} = \{Y \in \mathbb{G}_2(4, 2): Y \subset T \in \mathbb{G}_2(4, 3)\}$.
\end{proposition}
\begin{proof}
	By Theorem~\ref{T7}, $\mathcal{U}^{*} \subseteq \mathbb{G}_2(n, 2)$ is $1$-intersecting. As $n \ge 4$ and $|\mathcal{U}^{*}| = 2^{n-1}-1 = {{n-1} \brack {1}}_2$, the rest follows from Theorem~\ref{T6}. 
\end{proof}

Finding equidistant linear codes in $\mathbb{P}_2(n)$ of maximum size with constant distance of $2d$ for $d > 1$ is equivalent to finding maximal $d$-intersecting families in $\mathbb{G}_2(n, 2d)$ (Theorem~\ref{T7}). However, no example of construction of a $\lambda$-intersecting family in $\mathbb{G}_2(n, k)$ exist in the literature for $1 < \lambda < k-1$ to the best of our knowledge. We will discuss this issue in detail in Section~\ref{sec:6}.

%
%

\section{Equidistant Linear Codes in $\mathbb{P}_q(n)$ for $q > 2$}
\label{sec:6}
Linear codes in $\mathbb{P}_2(n)$ were first constructed in \cite{BEV}. While one of those constructions, viz. \emph{codes derived from a fixed basis} was generalized for any prime power $q$ \cite{PS}, the first novel construction of a linear code in $\mathbb{P}_q(n)$ for $q > 2$ was presented in \cite[Theorem~7]{BK2}. However, no equidistant linear code in $\mathbb{P}_q(n)$ has so far been constructed for $q > 2$.

\subsection{Existence of Equidistant Linear Codes in $\mathbb{P}_q(n)$}
\label{sec:6a}
Equidistant linear codes in $\mathbb{P}_2(n)$ with maximum size were shown to be symmetric designs in Section~\ref{sec:5} by exploiting the order of the underlying field. Similar technique does not apply for $q > 2$. We thus bring in the notion of Steiner triple systems to form equidistant linear codes out of certain $\lambda$-intersecting families in Grassmannians. The finding throws some light on the upper bound on size of equidistant linear codes in $\mathbb{P}_q(n)$ for all prime powers $q$.
\begin{theorem}
	\label{T9}
	An equidistant linear code $\mathcal{U} \subseteq \mathbb{P}_q(n)$ with constant distance of $2d$ and size $|\mathcal{U}|= 2^m$ exists if and only if there exists a $d$-intersecting uniform family $\mathcal{F} \subseteq \mathbb{G}_q(n, 2d)$ of size $|\mathcal{F}| = 2^m -1$.
\end{theorem}
\begin{proof}
	(Proof of $\Rightarrow$) According to Theorem~\ref{T7}, $\mathcal{U^{*}}$ is the desired $d$-intersecting family in $\mathbb{G}_q(n, 2d)$. \\
	(Proof of $\Leftarrow$) Say, a $d$-intersecting family $\mathcal{F} \subseteq \mathbb{G}_q(n, 2d)$ of size $(2^m-1)$ exists. Enumerate the members of $\mathcal{F}$ as $\mathcal{F} = \{X_1, X_2, \ldots, X_{2^m-1}\}$. It is easy to check that $2^m-1 \equiv 1, 3 \mod6$ for all positive integers $m \in \mathbb{Z}^{+}$. Thus, Theorem~\ref{T3} suggests that a collection $T \subseteq \binom{[2^m-1]}{3}$ of $3$-subsets of $[2^m-1]$ exists such that $T$ is an $STS(2^m-1)$.
	
	We define a set $\mathcal{U} := \mathcal{F} \cup \{\{0\}\}$ and a function $\sigma$ on $\mathcal{U}$ as:
	\begin{align}
	\sigma : \mathcal{U} \times \mathcal{U} &\longrightarrow \mathcal{U} \nonumber \\
	(Z, \{0\}), (\{0\}, Z) &\longmapsto Z \quad &\forall Z \in \mathcal{U} \nonumber \\
	(Z, Z) &\longmapsto \{0\} \quad &\forall Z \in \mathcal{U} \nonumber \\
	(X_i, X_j) &\longmapsto X_l \quad &\Longleftrightarrow \{i, j, l\} \in T.
	\label{E2}
	\end{align}
	Since $\{i, j, l\} \in T$ containing $\{i, j\}$ is unique for a fixed $\{i, j\} \in \binom{[2^m-1]}{2}$ by definition, mapping in \eqref{E2} is well defined, i.e. $\sigma$ is well defined. It is straightforward to verify that $\sigma$ is fit for acting as a linear addition on $\mathcal{U}$ (Definition~\ref{D1}). Hence, $\mathcal{U} \subseteq \mathbb{P}_q(n)$ is a linear code which is equidistant with constant distance $2d$ since $\mathcal{U^{*}} = \mathcal{F}$ is $d$-intersecting. As $|\mathcal{U}| = |\mathcal{F}| + 1 = 2^m$, the result follows. 
\end{proof}

\begin{corollary}
	\label{C2}
	Let $\mathcal{M}_{q,d}$ be the largest size of a $d$-intersecting family in $\mathbb{G}_q(n, 2d)$ for all $1 \le d \le \lfloor \frac{n}{2}\rfloor$. Suppose $\mathcal{U}$ is an equidistant linear code in $\mathbb{P}_q(n)$; if $\mathcal{U}$ has a constant distance of $2d$ then $|\mathcal{U}| \le 2^{\lfloor \log_2 (\mathcal{M}_{q,d} + 1)\rfloor}$. In general, $$|\mathcal{U}| \le 2^{\lfloor \log_2 (\mathcal{M}_q + 1)\rfloor},$$ where $\mathcal{M}_q := \max\limits_{1 \le d \le \lfloor n/2\rfloor} \mathcal{M}_{q,d}$.
\end{corollary}
\begin{proof}
	That $|\mathcal{U}| \le 2^{\lfloor \log_2 (\mathcal{M}_d + 1)\rfloor}$ if $2d$ is the constant distance of $\mathcal{U}$ follows directly from Theorem~\ref{T9}. The rest is proved by taking the maximum over all possible values of $d$. 
\end{proof}

The problem of determining the maximum size of an equidistant linear code in $\mathbb{P}_q(n)$ with constant distance of $2d$ is, therefore, equivalent to the problem of finding the largest $d$-intersecting family in $\mathbb{G}_q(n, 2d)$. In essence Theorem~\ref{T9} implies that the additional constraint of constant distance on linear codes in projective spaces decouples the translation invariance on the subspace distance metric. One final observation follows as a consequence of Theorem~\ref{T9}.
\begin{remark}
	\label{R4}
	A linear addition operation $\boxplus$ on an equidistant linear code $\mathcal{U} \subseteq \mathbb{P}_q(n)$ with size $|\mathcal{U}| = 2^s$ can be defined in exactly $N_s$ different ways, where $N_s$ represents the number of all possible distinct Steiner triple systems defined on the finite set $[2^s - 1]$ while separately counting all isomorphic systems. It must be noted that $N_s$ is independent of the value of $q$. 
\end{remark}
\subsection{Equidistant Linear Codes in $\mathbb{P}_q(n)$ of Size Greater than $2^n$}
\label{sec:6b}
We will henceforth use the following notations: $\mathbb{E}_q(n)$ and $\mathbb{E}_q(n,d)$ will represent the maximum size of an equidistant linear code in $\mathbb{P}_q(n)$ and the same when the constant distance is $2d$, respectively. Since the size of largest $d$-intersecting families in $\mathbb{G}_q(n, 2d)$ is yet unknown for $d > 1$, Corollary~\ref{C2} suggests that hitherto there is no recognized closed-form expression for $\mathbb{E}_q(n,d)$ or $\mathbb{E}_q(n)$ in general. However, we assert the following for $n=3$.
\begin{proposition}
	\label{P3}
	$\mathbb{E}_q(3) = 2^{\lfloor \log_2(q^2+q+2)\rfloor}$.
\end{proposition}
\begin{proof}
	The family $\mathcal{F} = \mathbb{G}_q(3, 2)$ is $1$-intersecting. Since $|\mathbb{G}_q(3, 2)| = {3 \brack 2}_q = q^2+q+1$, Theorem~\ref{T9} implies the rest. 
\end{proof}

Braun et al. conjectured the largest size of a linear code in $\mathbb{P}_2(n)$ to be $2^n$ \cite{BEV}. The existing examples of linear codes in $\mathbb{P}_q(n)$ are all of size at most $2^n$ \cite{BEV, PS, BK2}. Proposition~\ref{P3} argues that equidistant linear codes of size more than $2^n$ exist in $\mathbb{P}_q(n)$ for all prime powers $q > 3$. Table~\ref{tab:1} illustrates this point. It must be observed that the value of $\mathbb{E}_q(3)$ increases monotonically with $q$, the order of the underlying field.

\begin{table}[h!]
	\begin{center}
		\caption{Size of largest equidistant linear codes in $\mathbb{P}_q(3)$}
		\label{tab:1}       
		\begin{tabular}{c|c|c}
			\hline\noalign{}
			\textbf{Order of the field} & \textbf{Size of maximal family} & \textbf{Maximum size of a code}\\
			$q$ & $|\mathbb{G}_q(3, 2)|=q^2+q+1$ & $\mathbb{E}_q(3) = 2^{\lfloor\log_2(q^2+q+2)\rfloor}$ \\
			\noalign{}\hline\noalign{}
			2 & 7 & 8 \\
			3 & 13 & 8 \\
			4 & 21 & 16\\
			5 & 31 & 32\\
			7 & 57 & 32\\
			8 & 73 & 64\\
			9 & 91 & 64\\
			11 & 133 & 128\\
			13 & 183 & 128\\
			16 & 273 & 256\\
			17 & 307 & 256\\
			\noalign{}\hline
		\end{tabular}
	\end{center}
\end{table}

Clearly, $\mathbb{E}_q(n) = \mathbb{E}_q(n, 1)$ for $n=3$ and $4$; however, $\mathbb{E}_q(n) \ge \mathbb{E}_q(n, 1)$ in general. We at this point recall the $q$-analog of Erd{\"o}s-Ko-Rado theorem to prove a lower bound on $\mathbb{E}_q(n)$.
\begin{proposition}
	\label{P4}
	For all $n \ge 4$, $\mathbb{E}_q(n, 1) = 2^{\lfloor \log_2({\frac{q^{n-1}-1}{q-1}} + 1)\rfloor} \le \mathbb{E}_q(n)$. Equidistant linear codes in $\mathbb{P}_q(n)$ attaining this bound can always be constructed from a sunflower in $\mathbb{G}_q(n, 2)$ for any $n \ge 4$ and prime power $q$.
\end{proposition}
\begin{proof}
	The proof is by construction. Any intersecting family in $\mathbb{G}_q(n, 2)$ is $1$-intersecting by definition. Thus, by Theorem~\ref{T6} and Corollary~\ref{C2} any equidistant linear code $\mathcal{U} \subseteq \mathbb{P}_q(n)$ with constant distance of $2$ must obey the following:
	\begin{equation}
	\label{E3}
	|\mathcal{U}| \le 2^{\lfloor \log_2({{n-1} \brack 1}_q + 1)\rfloor}.
	\end{equation}
	Moreover, there always exists a $1$-intersecting sunflower $\mathbb{S}$ in $\mathbb{G}_q(n, 2)$ with size $|\mathbb{S}| = {{n-1} \brack 1}_q = \frac{q^{n-1}-1}{q-1}$. Any optimal equidistant linear code $\mathcal{U} := {\tilde{\mathbb{S}}} \cup \{\{0\}\}$ constructed from $\mathbb{S}$ for some $\tilde{\mathbb{S}} \subseteq \mathbb{S}$ with size $|\tilde{\mathbb{S}}| = 2^{\lfloor \log_2({{n-1} \brack 1}_q + 1)\rfloor} - 1$ attains the bound depicted in \eqref{E3}, proving the statement. 
\end{proof}
\begin{remark}
	\label{R5}
	For $n=4$ there is a known non-sunflower construction: $\mathcal{U}:= \tilde{\mathbb{S}} \cup \{\{0\}\}$, where $\tilde{\mathbb{S}}$ can be any subset of $\{Y \in \mathbb{G}_q(4, 2): Y \subset H \in \mathbb{G}_q(4, 3)\}$ with the desired size.
\end{remark}

It is easy to verify from \eqref{E3} that $\mathbb{E}_3(n, 1) \ge 2^{n+1}$ for all $n \ge 7$ (Table~\ref{tab:2}). Combining Proposition~\ref{P4} with Proposition~\ref{P3} leads us to the following conclusion.
\begin{corollary}
	\label{C3}
	For any prime power $q > 2$ there exist equidistant linear codes in $\mathbb{P}_q(n)$ with size at least $2^{n+1}$ for some $n \ge 3$.
\end{corollary}

\begin{table}[h!]
	\begin{center}
		\caption{Size of largest equidistant linear codes in $\mathbb{P}_3(n)$ with distance $2$}
		\label{tab:2}       
		\begin{tabular}{c|c|c}
			\hline\noalign{}
			\textbf{Dimension} & \textbf{Maximum size of a code} & \textbf{Size of coding space}\\
			$n$ & $\mathbb{E}_3(n, 1)$ & $|\mathbb{P}_3(n)|$ \\
			\noalign{}\hline\noalign{}
			3 & 8 & 28 \\
			4 & 8 & 212\\
			5 & 32 & 2664\\
			6 & 64 & 56632\\
			7 & 256 & 2052656\\
			8 & 1024 & 127902864\\
			\noalign{}\hline
		\end{tabular}
	\end{center}
\end{table}

The maximum size of a class of codes indicates how much of the entire coding space can be utilized by that scheme. For a code in $\mathbb{P}_q(3)$, size of the entire coding space is $|\mathbb{P}_q(3)| = \sum\limits_{i=0}^{3} |\mathbb{G}_q(3,i)| = 2(q^2+q+2)$, while according to Proposition~\ref{P3} the maximum size of an equidistant linear code in $\mathbb{P}_q(3)$ is at most $(q^2+q+2)$; thus, these codes can use at most half of the coding space. We will now establish that such class of linear codes are also guaranteed to exploit a certain minimum fraction of the underlying space. A result proved in \cite{RN} that was conjectured by Ramanujan is listed first.
\begin{theorem}[Ramanujan-Nagell's Theorem]
	\label{T10}
	A positive integer $x \in \mathbb{Z}^{+}$ satisfies the equation $x^2+7 = 2^m$ for some positive integer $m \in \mathbb{Z}^{+}$ if and only if $x = 1, 3, 5, 11$ or $181$.
\end{theorem}
\begin{lemma}
	\label{L9}
	For all prime powers $q$,
	\begin{equation*}
	\frac{1}{4} < \frac{\mathbb{E}_q(3)}{|\mathbb{P}_q(3)|} \le \frac{1}{2}.
	\end{equation*}
	The upper bound is reached if and only if $q = 2$ or $5$.
\end{lemma}
\begin{proof}
	There exists $m \in \mathbb{Z}^{+}$ for a fixed $q$ such that $2^{m-1} \le q^2+q+2 < 2^m$. Since $|\mathbb{P}_q(3)| = 2(q^2+q+2)$, we can say that $2^m \le |\mathbb{P}_q(3)| < 2^{m+1}$. Thus, both the bounds follow from Proposition~\ref{P3}. Equality on the upper bound is attained if and only if $q^2+q+2=2^{m-1}$. Theorem~\ref{T10} dictates that the only permissible solutions for $q$ are $2$ and $5$. 
\end{proof}

Proceeding in the same direction for higher values of $n$, we discover that any equidistant linear code in a projective space is able to exploit not more than half of the coding space. We make use of a few identities: for all prime powers $q$, ${n \brack {k-1}}_q < {n \brack k}_q$ for all $0 \le k <\lfloor \frac{n}{2}\rfloor$; and ${n \brack j}_q = {n \brack {n-j}}_q$ for all $0 \le j \le \lfloor \frac{n}{2}\rfloor$.
\begin{lemma}
	\label{L10}
	For all prime powers $q$, $\frac{\mathbb{E}_q(n)}{|\mathbb{P}_q(n)|} \le \frac{1}{2}$. Equality occurs only if $n=3$ and $q=2$ or $5$.
\end{lemma}
\begin{proof}
	We begin with the case $n = 4d$ for some $d \in \mathbb{Z}^{+}$. Consider any equidistant linear code $\mathcal{U} \subseteq \mathbb{P}_q(4d)$ with constant distance of $2d$. For every $X \in \mathbb{G}_q(4d, 2d)$ such that $X \in \mathcal{U^{*}}$, there exists $X^{\perp} \in \mathbb{G}_q(4d, 2d)$, the orthogonal complement of $X$ with respect to the standard inner product in $\mathbb{F}_q^{4d}$ such that $X \cap X^{\perp} = \{0\}$, i.e. $X^{\perp} \notin \mathcal{U^{*}}$. Thus, $\mathbb{E}_q(4d) - 1 \le \frac{1}{2} |\mathbb{G}_q(4d, 2d)|$ or in other words, $\mathbb{E}_q(4d) < \frac{1}{2} |\mathbb{P}_q(4d)|$. \\
	
	Next, we consider the case $n \ne 4d$. By definition, $\mathbb{E}_q(n) \le |\mathbb{G}_q(n, 2d)|+1$ for some $1 \le d \le \lfloor \frac{n}{2}\rfloor$. Since $|\mathbb{G}_q(n, 2d)| = {n \brack {2d}}_q = {n \brack {n-2d}}_q = |\mathbb{G}_q(n, n-2d)|$, for all permissible values of $n \ge 3$ we have $|\mathbb{P}_q(n)| = \sum\limits_{k=0}^{n} |\mathbb{G}_q(n, k)| \ge 2(1+|\mathbb{G}_q(n, 2d)|)$. Equality holds only if $2d = n-1$ and $d=1$, i.e. only if $n=3$. The rest follows from Lemma~\ref{L9}. 
\end{proof}

However, the upper bound in Lemma~\ref{L10} gets weaker as either of $n$ and $q$ grows. Table~\ref{tab:2} gives an idea to how fast the ratio shrinks with increasing $n$. In retrospect, Corollary~\ref{C2} suggests that $\frac{\mathbb{E}_q(n)}{|\mathbb{P}_q(n)|}$ is much smaller than $0.5$ as a $d$-intersecting family in $\mathbb{G}_q(n, 2d)$ is likely to be too small compared to the projective space $\mathbb{P}_q(n)$. We conclude this section with an observation about the largest $d$-intersecting families in $\mathbb{G}_2(n, 2d)$.
\begin{proposition}
	\label{P5}
	If $\mathcal{M}_{q, d}$ is the maximum size of a $d$-intersecting family in $\mathbb{G}_q(n, 2d)$ then for all $n \ge 4$,
	\begin{equation*}
	\mathcal{M}_{2, d} < 2^n-1.
	\end{equation*}
\end{proposition}
\begin{proof}
	Suppose there exists a $d$-intersecting family $\mathcal{F} \subseteq \mathbb{G}_2(n, 2d)$ such that $|\mathcal{F}| = 2^n-1$. An equidistant linear code $\mathcal{U} := \mathcal{F} \cup \{\{0\}\}$ can thus be constructed as per Theorem~\ref{T9} with size $|\mathcal{U}| = 2^n$, contradicting Theorem~\ref{T8}. 
\end{proof}

\section{Conclusions}
\label{sec:7}
We have studied linear codes in projective spaces with the additional constraint that any pair of codewords are at a fixed distance apart. Our findings indicate that an equidistant linear code in $\mathbb{P}_q(n)$ must have an even constant distance and the nontrivial part of a linear code with fixed distance $2d$ is necessarily a subset of the Grassmannian $\mathbb{G}_q(n, 2d)$ with all pairwise intersections of constant dimension $d$. We have further proved that an equidistant linear code with distance $2d$ can be constructed from any $d$-intersecting family in $\mathbb{G}_q(n, 2d)$, thereby establishing that the property of constant distance decouples the translation invariance on the subspace distance metric. An upper bound on the size of such class of linear codes is derived by us which elaborates that the maximum size of an equidistant linear code in $\mathbb{P}_q(n)$ is $2^n$ only when $q$, the order of the underlying field is $2$. We have also identified the nontrivial parts of optimal equidistant linear codes in $\mathbb{P}_2(n)$ to be analogous to the Fano plane (for $n=3$) and sunflowers (for $n \ge 4$). First example of a linear code in $\mathbb{P}_q(n)$ with size more than $2^n$ is presented in this work.

The expression for maximum size of an equidistant linear code in $\mathbb{P}_q(n)$ is, however, unknown for all prime powers $q > 2$ and remains an open problem for future research. In particular, we ask the following.
\begin{problem}
Find the maximum size of a $d$-intersecting family in $\mathbb{G}_q(n, 2d)$ for $1 \le d \le \lfloor \frac{n}{2}\rfloor$, for all $n \ge 4$ and all prime powers $q > 2$; identify the maximal families.
\end{problem}
\begin{problem}
Find the values of $n \ge 5$ and prime powers $q$ for which $\mathbb{E}_q(n) = \mathbb{E}_q(n, 1)$.
\end{problem}
\paragraph{Acknowledgements}
The author would like to thank Prof. Navin Kashyap and Prof. Arvind Ayyer for their valuable inputs.

\bibliographystyle{ieeetr}
\bibliography{Equidistant}

\end{document}